\documentclass[english]{article}
\usepackage[T1]{fontenc}
\usepackage[latin9]{inputenc}
\usepackage{xcolor}
\usepackage{amsmath}
\usepackage{amsthm}
\usepackage{amssymb}
\PassOptionsToPackage{normalem}{ulem}
\usepackage{ulem}

\makeatletter

\providecolor{lyxadded}{rgb}{0,0,1}
\providecolor{lyxdeleted}{rgb}{1,0,0}

\DeclareRobustCommand{\lyxsout}[1]{\ifx\\#1\else\sout{#1}\fi}

\newcommand{\lyxaddress}[1]{
	\par {\raggedright #1
	\vspace{1.4em}
	\noindent\par}
}
\theoremstyle{plain}
\newtheorem{thm}{\protect\theoremname}
\theoremstyle{plain}
\newtheorem{lem}[thm]{\protect\lemmaname}
\ifx\proof\undefined
\newenvironment{proof}[1][\protect\proofname]{\par
	\normalfont\topsep6\p@\@plus6\p@\relax
	\trivlist
	\itemindent\parindent
	\item[\hskip\labelsep\scshape #1]\ignorespaces
}{%
	\endtrivlist\@endpefalse
}
\providecommand{\proofname}{Proof}
\fi

\@ifundefined{date}{}{\date{}}
\usepackage[all]{xy}

\usepackage{cite}

\usepackage{multirow}

\def\frontmatter@abstractheading{}

\usepackage{babel}

  \providecommand{\lemmaname}{Lemma}
  
\providecommand{\theoremname}{Theorem}

\makeatother

\usepackage{babel}
\providecommand{\lemmaname}{Lemma}
\providecommand{\theoremname}{Theorem}

\begin{document}
\title{Epistemic Odds of Contextuality\\in Cyclic Systems}
\author{Ehtibar N.\ Dzhafarov\textsuperscript{1}, Janne V.\ Kujala\textsuperscript{2},
V\'ictor H.\ Cervantes\textsuperscript{1}}
\maketitle

\lyxaddress{\begin{center}
\textsuperscript{1}Purdue University, \textsuperscript{2}University
of Turku
\par\end{center}}
\begin{abstract}
Beginning with the Bell theorem, cyclic systems of dichotomous random
variables have been the object of many foundational findings in quantum
mechanics. Here, we ask the question: if one chooses a cyclic system
``at random'' (uniformly within the hyperbox of all possible systems
with given marginals), what are the odds that it will be contextual?
We show that the odds of contextuality rapidly tend to zero as the
size of the system increases. The result is based on the Contextuality-by-Default
theory, in which we do not have to assume that the systems are subject
to the no-disturbance/no-signaling constraints.
\end{abstract}

\section{Introduction}

Cyclic systems of dichotomous random variables have played a prominent
role in contextuality research. Suffice it to say that they are the
object of the celebrated Bell theorem \cite{Bell1964,Bell1966,CHSH1969,Fine1982},
as well as the Leggett-Garg theorem \cite{LeggGarg1985,Bacciagaluppi2015,KoflerBrukner2013,SuppesZanotti1981},
Klyachko-Can-Binicio\u{g}lu-Shumovsky theorem \cite{KCBS2008,Lapkiewicz2015,KujDzhLar2015},
and many other results. In this paper we present a simple proof of
the following proposition: the epistemic (Bayesian) probability that
a randomly chosen cyclic system of dichotomous random variables is
contextual tends to zero as its rank $n$ increases. The terms in
this statement are to be rigorously defined later, but the gist is
as follows. Systems of random variables representing measurements
or hypothetical physical events can be classified into contextual
and noncontextual. If a system is of a special kind, called cyclic,
it is represented by a point within an $n$-dimensional hyperbox $\mathbb{B}$
whose edges are determined by the individual (marginal) distributions
of the random variables the system contains. We consider these distributions
fixed, and different meanings of a ``randomly chosen'' cyclic system
correspond to different distributions of points within the hyperbox
$\mathbb{B}$. Here, we assume this distribution to be uniform. A
part of the hyperbox $\mathbb{B}$ forms a noncontextuality polytope
$\mathbb{K}$ consisting of all points representing noncontextual
cyclic systems. The \emph{epistemic probability} of choosing a contextual
system is then
\begin{equation}
\epsilon=1-\frac{\mathsf{vol}\left(\mathbb{K}\right)}{\mathsf{vol}\left(\mathbb{B}\right)},
\end{equation}
where $\mathsf{vol}$ stands for Euclidean hypervolume. Termed differently,
$\epsilon$ is the Bayesian probability of contextuality, with uniform
prior. The precise value of $\epsilon$ depends on the marginal distributions
that define the hyperbox $\mathbb{B}$, but we show that $\epsilon\leq2^{n-1}/n!$,
which tends to zero as $n$ increases. The paper draws on the recent
detailed analysis of cyclic systems given in \cite{DKC2020}.

\section{Definitions}

Our analysis is based on the Contextuality-by-Default (CbD) theory
\cite{DzhCerKuj2017,DzhKujFoundations2017,KujDzhLar2015,KujDzhMeasures}.
A \emph{cyclic system of rank} $n=2,3,\ldots$, is a system
\begin{equation}
\mathcal{R}_{n}=\left\{ \left\{ R_{i}^{i},R_{i\oplus1}^{i}\right\} :i=1,\ldots,n\right\} ,\label{eq:cyclic_n}
\end{equation}
where $\oplus1$ is the cyclic shift $1\mapsto2$, $\dots$, $n-1\mapsto n$,
$n\mapsto1$, and $\left\{ R_{i}^{i},R_{i\oplus1}^{i}\right\} $ are
pairs of jointly distributed random variables. We will assume here
that all $R_{j}^{i}$ are dichotomous, $0/1$-variables (although
any other labeling will be equally acceptable). The matrix below presents
an example of a cyclic system:
\begin{equation}
\begin{array}{|c|c|c|c||c}
\hline R_{1}^{1} & R_{2}^{1} &  &  & c=1\\
\hline  & R_{2}^{2} & R_{3}^{2} &  & c=2\\
\hline  &  & R_{3}^{3} & R_{4}^{3} & c=3\\
\hline R_{1}^{4} &  &  & R_{4}^{4} & c=4\\
\hline\hline q=1 & q=2 & q=3 & q=4 & \mathcal{R}_{4}
\end{array}\label{eq: R4}
\end{equation}
This is a cyclic system of rank 4, describing, e.g., the object of
the best known version of the Bell theorem \cite{Bell1966,CHSH1969}.
The columns of the matrix correspond to properties $q$ being measured,
denoted by the subscripts of the variables. Thus, in the target application
of the Bell theorem, $q=1$ and $q=3$ represent Alice's settings,
while $q=2$ and $q=4$ represent Bob's settings. We generically refer
to $q=j$ in $R_{j}^{i}$ as the \emph{content} of this random variable.
The rows of the matrix correspond to \emph{contexts} in which the
random variables are pairwise recorded, denoted by their superscripts.
So, a random variable in a system is uniquely identified by its content
and its context.

If any two content-sharing (i.e., equally subscripted, measuring the
same property) random variables are identically distributed, i.e.,
if $\left\langle R_{j}^{i}\right\rangle =\left\langle R_{j}^{i'}\right\rangle $
for any $i,i',j$ for which $R_{j}^{i}$ and $R_{j}^{i'}$ exist,
the system is said to be \emph{consistently connected}. This is the
CbD term for compliance with the no-disturbance/no-signaling constraint.
Unlike most approaches to contextuality (an exception being \cite{AmaralDuarteOliveira2018}),
CbD does not need this assumption, so the cyclic systems here are
generally \emph{inconsistently connected}.

The definition of (non)contextuality is based on the notion of a \emph{coupling}.
In system (\ref{eq:cyclic_n}), any two context-sharing random variables
are jointly distributed, but any two random variables belonging to
different contexts are \emph{stochastically unrelated}. The system
$\mathcal{R}_{n}$ as a whole therefore is not jointly distributed.
A coupling of $\mathcal{R}_{n}$ is a set of jointly distributed random
variables (hence, a random variable in its own right)
\begin{equation}
S=\left\{ S_{j}^{i}:j=i,i\oplus1;i=1,\ldots,n\right\} 
\end{equation}
such that 
\begin{equation}
\left[\begin{array}{c}
\left\langle S_{i}^{i}\right\rangle =\left\langle R_{i}^{i}\right\rangle =p_{i}^{i}\\
\left\langle S_{i\oplus1}^{i}\right\rangle =\left\langle R_{i\oplus1}^{i}\right\rangle =p_{i\oplus1}^{i}\\
\left\langle S_{i}^{i}S_{i\oplus1}^{i}\right\rangle =\left\langle R_{i}^{i}R_{i\oplus1}^{i}\right\rangle =p_{i,i\oplus1}
\end{array}\right],i=1,\ldots,n.\label{eq:expectations}
\end{equation}
The system $\mathcal{R}_{n}$ is \emph{noncontextual} if it has a
coupling $S$ in which any two content-sharing random variables coincide
with maximal possible probability. It is easily seen that this means
\begin{equation}
\left\langle S_{i}^{i}S_{i}^{i\ominus1}\right\rangle =\min\left(p_{i}^{i},p_{i}^{i\ominus1}\right)=p^{i,i\ominus1},i=1,\ldots,n,\label{eq:max couplings}
\end{equation}
where $\ominus1$ is the inverse of $\oplus1$. If such a coupling
does not exist, the system is \emph{contextual}. The intuition is
that the contexts in this case ``force'' the content-sharing variables
to be more dissimilar than they can be if taken in isolation. In the
particular case of consistently connected systems, $p_{i}^{i}=p_{i}^{i\ominus1}$,
and one can say, with a slight abuse of language, that in contextual
systems the contexts prevent the content-sharing random variables
from being ``the same.'' This is, essentially, the traditional understanding
of contextuality \cite{Dzh2019,DzhKujFoundations2017}. There is a
simple closed-form criterion of (non)contextuality proved in \cite{KujDzhProof2016}
(and reduced to one proved in \cite{Araujoetal2013} in the special
case of consistently connected systems). In this paper, however, we
make no use of this criterion.

\section{Main Result}

To define epistemic probabilities, one needs a principled way of placing
a system within a space of systems. Here, we follow the scheme we
used in \cite{DKC2020} to define a noncontextuality polytope and
measures of (non)contextuality. Due to the prominent role of \cite{DKC2020}
in our reasoning, we explain notation correspondences with that paper
in several subsequent footnotes. We begin by introducing three probability
vectors.\footnote{The probability vectors $\mathbf{a},\mathbf{b},\mathbf{c}$ below
correspond, respectively, to $\mathbf{p_{l}},\mathbf{p_{b}},\mathbf{p_{c}}$
in \cite{DKC2020}.} With reference to (\ref{eq:expectations}), denote
\begin{equation}
\mathbf{a}=\left(1,p_{1}^{1},p_{2}^{1},\ldots,p_{n}^{n},p_{1}^{n}\right)^{\intercal},
\end{equation}
and
\begin{equation}
\mathbf{b}=\left(p_{12},p_{23},\ldots,p_{n-1,n},p_{n1}\right)^{\intercal}.
\end{equation}
With reference to (\ref{eq:max couplings}), denote
\begin{equation}
\mathbf{c}=\left(p^{1n},p^{21},\ldots,p^{n-1,n-2},p^{n,n-1}\right)^{\intercal}.
\end{equation}
A system represented by $\left(\mathbf{a},\mathbf{b},\mathbf{c}\right)^{\intercal}$
is noncontextual if and only if there is a vector $\mathbf{h}\ge0$
(componentwise) such that 
\begin{equation}
\mathbf{M}\mathbf{\mathbf{h}}=\mathbf{\left(\mathbf{a},\mathbf{b},\mathbf{c}\right)^{\intercal}},
\end{equation}
where $\mathbf{M}$ is an incidence (0/1) matrix whose detailed description
is given in \cite{DzhCerKuj2017,KujDzhMeasures,DKC2020}. For any
set of fixed 1-marginals (which means fixed vectors $\mathbf{a}$
and $\mathbf{c}$), we call the convex polytope
\begin{equation}
\mathbb{K}=\left\{ \left.\mathbf{b}\right|\exists\mathbf{h}\geq0:\mathbf{M}\mathbf{\mathbf{h}}=\mathbf{\left(\mathbf{a},\mathbf{b},\mathbf{c}\right)^{\intercal}}\right\} 
\end{equation}
the \emph{noncontextuality polytope}. We remark in passing that the
$L_{1}$-distance between a point $\mathbf{b}$ and the surface of
the polytope $\mathbb{K}$ is a natural measure of contextuality (if
$\mathbf{b}$ is outside $\mathbb{K}$) or noncontextuality (if $\mathbf{b}$
is in $\mathbb{K}$) of the system represented by $\mathbf{b}$ \cite{DKC2020,KujDzhMeasures}.

It is easily seen that 
\begin{equation}
\max\left(0,p_{i}^{i}+p_{i\oplus1}^{i}-1\right)\leq p_{i,i\oplus1}\leq\min\left(p_{i}^{i},p_{i\oplus1}^{i}\right),i=1,\ldots,n,
\end{equation}
whence $\mathbb{K}$ is inscribed in the hyperbox
\begin{equation}
\mathbb{B}=\prod_{i=1}^{n}\left[\max\left(0,p_{i}^{i}+p_{i\oplus1}^{i}-1\right),\min\left(p_{i}^{i},p_{i\oplus1}^{i}\right)\right].\label{eq: rectangle}
\end{equation}
Let us agree to call the vertices of $\mathbb{B}$ \emph{odd} or \emph{even}
depending on whether its coordinates contain, respectively, an odd
or even number of left endpoints of the intervals in (\ref{eq: rectangle}).
We need the following two results from \cite{DKC2020}.
\begin{lem}[Lemma 13 in \cite{DKC2020}]
\label{lm:KDC2020} Every even vertex of $\mathbb{B}$ belongs to
$\mathbb{K}$ (i.e., represents a noncontextual system).\footnote{\label{fn:P_b}In \cite{DKC2020}, where we also consider other polytopes,
the noncontextuality polytope $\mathbb{K}$ is denoted by $\mathbb{P}_{\mathbf{b}}$
or, if the values of the random variables are encoded as $-1/+1$
rather than 0/1, by $\mathbb{E}_{\mathbf{b}}$. The hyperbox $\mathbb{B}$,
with the $-1/+1$ encoding of values, is denoted by $\mathbb{R}_{\mathbf{b}}$.}
\end{lem}
\begin{lem}[Lemma 10 in \cite{DKC2020}]
\label{lem:KDC2020a} If all 1-marginal probabilities are $1/2$,
then $\mathbb{K}$ is the $n$-demicube whose vertices are even vertices
of $\mathbb{B}$.\footnote{Note that the system in this lemma is a special case of a consistently
connected system. In reference to (\ref{eq: rectangle}), $\mathbb{B}$
in this case is hypercube $\left[0,1/2\right]^{n}$. In \cite{DKC2020},
it corresponds to hypercube $\left[-1,1\right]^{n}$ denoted by $\mathbb{C_{\mathbf{b}}}$.}
\end{lem}
We are ready now to make our main observation. We define the epistemic
probability of a system falling within a Lebesque-measurable subset
$\mathbb{S}$ of $\mathbb{B}$ as the ratio of their Lebesque measures
(in our case, Euclidean volumes of polytopes).
\begin{thm}
\label{thm: main}For any cyclic system of rank $n$,
\begin{equation}
\epsilon=1-\frac{\mathsf{vol}\mathbb{\mathbb{K}}}{\mathsf{vol}\mathbb{B}}\leq\frac{2^{n-1}}{n!}.
\end{equation}
This upper bound is tight.
\end{thm}
\begin{proof}
By Lemma \ref{lm:KDC2020}, since $\mathbb{K}$ is convex, it contains
the polytope $\mathbb{D}$ formed by the even vertices of $\mathbb{B}$.
This polytope is an $n$-demibox.\footnote{We use this term as a straightforward generalization of $n$-demicube.}
Within $\mathbb{B}$, it is separated from any odd vertex of $\mathbb{\mathbb{B}}$
by the corner formed by this vertex and the endpoints of the sides
of $\mathbb{B}$ emanating from this vertex. The volume of this corner
is $L_{1}\ldots L_{n}/n!$, where $L_{i}$ is the length of the $i$th
side in (\ref{eq: rectangle}). There are $2^{n-1}$ such corners,
whence
\[
\mathsf{vol}\mathbb{D}=\mathsf{vol}\mathbb{B}-2^{n-1}\frac{L_{1}\ldots L_{n}}{n!}.\tag{*}
\]
Since $\mathsf{vol}\mathbb{B}=L_{1}\ldots L_{n}$, the epistemic probability
for a point randomly chosen within $\mathbb{B}$ to fall within $\mathbb{D}$
is
\[
\frac{\mathsf{vol}\mathbb{D}}{\mathsf{vol}\mathbb{B}}=1-\frac{2^{n-1}}{n!}.\tag{**}
\]
The upper bound stated in the theorem now follows from $\mathsf{vol}\mathbb{\mathbb{K}}\geq\mathsf{vol}\mathbb{D}$.
By Lemma \ref{lem:KDC2020a}, in the case of uniform 1-marginals $\mathbb{K}$
coincides with $\mathbb{D}$. This proves that the upper bound is
tight.
\end{proof}
Although we fixed the 1-marginals in the proof, we see that the epistemic
probability is bounded by an expression that only depends on $n$.
The convergence of this upper bound to zero is quite fast: 
\begin{equation}
\begin{array}{c|c|c|c|c|c|c|c|c|c}
n & 2 & 3 & 4 & 5 & 10 & 15 & 20 & 50 & \ldots\\
\hline \epsilon\leq & 1 & 6.67 & 3.34 & 1.34 & 1.42 & 1.26 & 2.16 & 1.86\\
 &  & \times10^{-1} & \times10^{-1} & \times10^{-1} & \times10^{-4} & \times10^{-8} & \times10^{-13} & \times10^{-50} & \ldots
\\\hline \end{array}
\end{equation}
Depending on the 1-marginals, $\epsilon$ can be much smaller than
$2^{n-1}/n!$.

One could also consider the epistemic probability $\widetilde{\epsilon}$
of choosing a contextual system in the $3n$-dimensional space formed
by all vectors 
\begin{equation}
\widetilde{\mathbf{p}}=\left(\left\langle R_{i}^{i}\right\rangle ,\left\langle R_{i}^{i}R_{i\oplus1}^{i}\right\rangle ,\left\langle R_{i\oplus1}^{i}\right\rangle :i=1,\ldots,n\right).
\end{equation}
Because the upper bound for $\epsilon$ in Theorem \ref{thm: main}
does not depend on 1-marginals, this upper bound also bounds $\widetilde{\epsilon}$,
irrespective of the epistemic distribution of choices of the 1-marginals
(provided the conditional probability of contextuality, given the
1-marginals, is defined as above).

\section{Conclusion}

Our finding is surprising, as it shows that complexity and contextuality
may very well be antagonists. Whether this has deeper interpretational
consequences depends on how much it can be generalized beyond the
class of cyclic systems. One problem is that size of a system is not
a well-defined concept outside specially defined classes of systems.
For cyclic systems, their rank $n$ determines simultaneously the
number of contexts ($n$), the number of contents ($n$), and the
number of random variables ($2n$). Generally, however, the number
of contents and contexts can be incremented independently, and it
is easy to see that our result will not always hold. Consider, e.g.,
a system with two contents and increasing number $n$ of contexts.
It can be shown that the epistemic probability with which such a system
is contextual generally does not decrease with increasing $n$ (e.g.,
within the class of consistently connected systems this epistemic
probability is 1 for all $n\geq2$). Even if we define the size of
a system as the rank of its largest cyclic subsystem, our result still
will not be generalized automatically: as shown in \cite{DKC2020},
a system whose cyclic subsystems are all noncontextual may very well
be contextual (although the epistemic probability of this has not
been investigated). Further work is needed.

What can be said about cyclic systems of random variables that are
not dichotomous? CbD requires that each random variables in an initial
system be replaced with a set of jointly distributed dichotomizations
thereof, and only then subjected to contextuality analysis \cite{DzhCerKuj2017}.
However, a cyclic system thus dichotomized is no longer cyclic. In
the case of categorical random variables with unordered sets of values,
we form all possible dichotomizations, and then we have a simple necessary
condition for noncontextuality, given by the \emph{nominal dominance
theorem} \cite{DzhCerKuj2017}. Using this condition, our computations
show that for cyclic systems the epistemic probability of contextuality
\emph{increases} with the number of unordered values of the random
variables. This, however, is not an easily interpretable result, because
as the set of possible values of random variables increases in cardinality,
it is progressively less feasible to treat it as completely unordered,
and it becomes impossible when the cardinality is infinite. For ordered/structured
sets of values the idea of all possible dichotomizations is no longer
justifiable, and the nominal dominance theorem no longer applies (see
the discussion in the concluding section of \cite{DzhCerKuj2017}).
Further work is needed.

One implication is obvious, however: insofar as one is concerned with
cyclic systems of dichotomous random variables, unless one is guided
by a predictive theory, one is unlikely to stumble upon a contextual
system of a sufficiently large size. Of course, quantum mechanics
is such a predictive theory, which is why we know of the existence
of contextual cyclic systems (although even there, most of experimental
work is confined to cyclic systems of ranks not exceeding 5). However,
our finding poses a serious problem for attempts to seek contextuality
outside quantum mechanics, where such a predictive theory may not
exist.

\end{document}